\documentclass[12pt]{article}
\usepackage{amsmath}
\usepackage{geometry}
\usepackage{algorithm,float}
\usepackage[noend]{algpseudocode}
\usepackage{graphicx}
\usepackage{caption}
\usepackage{setspace}
\usepackage[font={small,it}]{caption}
\usepackage{enumitem}
\usepackage{amssymb}
\usepackage{algpseudocode}
\usepackage{algorithm}
\usepackage{bm}
\usepackage{color}
\usepackage{epstopdf}
\usepackage{pdflscape, amsthm}
\usepackage{subcaption}
\usepackage{multirow,ulem}

\newtheorem{theorem}{Theorem}
\newtheorem{condition}{Condition}
\newtheorem{assumption}{Assumption}

\newcommand{\bmX}{\bm{X}}
\newcommand{\bmx}{\bm{x}}
\newcommand{\bmW}{\bm{W}}
\newcommand{\bmw}{\bm{w}}
\newcommand{\bmU}{\bm{U}}
\newcommand{\bmu}{\bm{u}}

\title{Bayesian Regularization of Gaussian Graphical Models with Measurement Error}
\author{Michael Byrd, Linh Nghiem, and Monnie McGee}

\begin{document}
\maketitle
\begin{abstract}
We consider a framework for determining and estimating the conditional pairwise relationships of variables when the observed samples are contaminated with measurement error in high dimensional settings. Assuming the true underlying variables follow a multivariate Gaussian distribution, if no measurement error is present, this problem is often solved by estimating the precision matrix under sparsity constraints. However, when measurement error is present, not correcting for it leads to inconsistent estimates of the precision matrix and poor identification of relationships. We propose a new Bayesian methodology to correct for the measurement error from the observed samples. This Bayesian procedure utilizes a recent variant of the spike-and-slab Lasso to obtain a point estimate of the precision matrix, and corrects for the contamination via the recently proposed Imputation-Regularization Optimization procedure designed for missing data. Our method is shown to perform better than the naive method that ignores measurement error in both identification and estimation accuracy. To show the utility of the method, we apply the new method to establish a conditional gene network from a microarray dataset.  
\end{abstract}

\section{Introduction}
A core problem in statistical inference is estimating the conditional relationship among random variables.  Naturally, a full description of the underlying connections among the numerous random variables is valuable information across many disciplines, such as in biology where the relationships among hundreds of genes involved in a metabolic process is desired to be uncovered.  In fact, under the assumption that the variables follow a multivariate Gaussian distribution, the inverse covariance matrix, known as the precision matrix, characterizes conditional dependence between two dimensions. This is accomplished by noting that if an element of the precision matrix is 0, then the two variables are conditionally independent; see \cite{lauritzen1996graphical} for a review.  This setting, often referred to as a Gaussian graphical model, is where our analysis takes place.

Estimating the precision matrix is a difficult task when the number of observations $n$ is often much less than the dimension of the features $d$. A naive approach is to estimate the precision matrix by the inverse of the empirical covariance matrix; this estimate, however, is known to perform poorly and is ill-posed when $n < d$ \cite{johnstone2001distribution}.  The common approach is to assume that the precision matrix is sparse \cite{dempster1972covariance}; that is, we assume the precision matrix's off-diagonal elements are mostly 0. As a result, most pairs of variables are conditionally independent. The sparsity assumption has led to different lines of research with regularized models to estimate the precision matrix.  While one approach utilizes a sparse regression technique that estimates the precision by iteratively regressing each variable on the remaining variables, for instance \cite{khare2015convex}, we instead focus on the direct likelihood approach. The direct likelihood approach optimizes the full likelihood function with an element-wise penalty on the precision matrix; common examples being graphical lasso \cite{friedman2008sparse}, CLIME \cite{cai2011constrained}, and TIGER \cite{liu2017tiger}.  We utilize a recent Bayesian optimization procedure, called BAGUS, that relies on optimization performed by the EM-algorithm, which was shown to have desirable theoretical properties, including consistent estimation of the precision matrix and selection consistency of the conditional pair-wise relationships \cite{gan2018bayesian}.

There are many practical issues associated with Gaussian graphical models, such as hyperparameter tuning \cite{yuan2007model}, missing data \cite{liang2018imputation}, and repeated trials \cite{tan2016replicates}, which practitioners need to adjust for a successful analysis.  We address another practical issue involved with these models, measurement error.  Measurement error occurs when the  variables of interest are not observed directly; instead, the observations are the desired variables that have been additionally perturbed with noise from some measurement process. This happens when, for instance, an inaccurate device is used to measure some sort of health metric. Measurement error models have been studied extensively for classical settings such as density deconvolution and regression \cite{carroll2006measurement}, but, to our knowledge, have not yet been well studied in the context of Gaussian graphical models, especially in high dimensional setting.  

We propose a Bayesian methodology to correct for measurement error in estimating a sparse precision matrix; our new method extends the optimization procedure of \cite{gan2018bayesian}.  While directly incorporating the estimate of the uncontaminated variable is possible, we find the incorporation of the imputation-regularization technique of \cite{liang2018imputation} to provide more desirable results.  Our procedure imputes the mismeasured random variables, then performs BAGUS on this imputation; these steps are performed for a small number of cycles, requiring more computation but giving better results than the naive estimator.  We prove consistency of the estimated precision matrix with the imputed procedure, and illustrate the performance in a simulation study.  We conclude with an application to a microarray data.

\section{Contaminated Gaussian Graphical Models}
\label{section:ContaminatedModel}

Given a $d$-dimensional random vector, $\bmx = \{x^1, \ldots, x^d\}$, the conditional dependence of two variables $x^i$ and $x^j$, for any pair $(i,j)$ with $1\leq i \leq j \leq d$, given all the remaining variables is of interest. This conditional dependence structure is usually represented by an undirected graph $G = (V,E)$, where $V = \{1, \ldots, d\}$ is the set of nodes and $E \subseteq V \times V = \{(1,1),(1,2),\ldots, (d,d)\}$ is the set of edges \cite{lauritzen1996graphical}. In this representation, the two variables $x^i$ and $x^j$ are conditionally independent if there is no edge between node $i$ and  node $j$. 

If the vector $\bmx$ follows a multivariate normal distribution with mean $\bm0$ and covariance matrix $\bm{\Sigma}_x$, $\bmx \sim N_d(\bm{0},\bm\Sigma_x)$, every edge corresponds to a non-zero entry in the precision matrix $\bm\Omega_x = \bm\Sigma_x^{-1}$, see \cite{lauritzen1996graphical}. The model in this scenario is often known as a Gaussian graphical model.  In the high dimensional setting, the set of edges are usually assumed to be sparse, meaning that only a few pairs $(x^i, x^j)$ are conditionally dependent. In the Gaussian case, this assumption implies only a few off-diagonal entries of $\bm \Omega_x$ are non-zero.
	
	When measurement error is present, denote $\bmU = (\bmu_1 , \ldots , \bmu_n)^T$ as measurement errors that are independent from data $\bmX = (\bmx_1 , \ldots , \bmx_n)^T$. For $i = 1 , \ldots , n$, the amount of measurement error is drawn from another multivariate normal distribution with mean $\bm0$ and covariance matrix $\bm{\Sigma_u}$, $\bmu_i \sim N_d(\bm0 , \bm{\Sigma_u}).$  We assume $\bm{\Sigma_u}$ to be diagonal, and hence the amount of measurement error on each variable is uncorrelated. We make a common assumption that $\bm{\Sigma_u}$ is known or estimable from ancillary data, such as replicate measurements. The contaminated variables $\bmw = \bmx + \bmu$ in general have a different conditional dependence structure from that of $\bmx$. Indeed, the covariance and precision matrix of $\bmw$ is given by
	\[ \bm{\Sigma}_w = \bm{\Sigma_x} + \bm{\Sigma}_u
	\]
	and 
	\begin{equation} \bm{\Omega}_w = \bm{\Sigma}_w^{-1} = \left(\bm{\Sigma_x} + \bm{\Sigma}_u\right)^{-1} = \bm{\Omega}_x - \bm{\Omega}_x (\bm{I} + \bm\Sigma_u\bm\Omega_x)^{-1} \bm{\Sigma}_u \bm\Omega_x
	\label{eq:Omega_w},
	\end{equation}
	respectively; here, $\bm{I}$ denotes the $d \times d$ identity matrix. Equation \eqref{eq:Omega_w} follows from the Kailath variant formula in \cite{petersen2008matrix}. Furthermore, equation \eqref{eq:Omega_w} suggests that $\bm\Omega_w$ and $\bm\Omega_x$ are equal if the product $\bm{\Omega}_x (\bm{I} + \bm\Sigma_u\bm\Omega_x)^{-1} \bm\Sigma_u\bm\Omega_x$ is equal to a zero matrix. This is generally not the case when the matrix $\bm\Sigma_u$ is not zero. 

	Suppose the data consist of \textit{iid} observations $\bmw_1 , \ldots , \bmw_n$, where $\bmw_i = \bmx_i + \bmu_i,~ i = 1, \ldots, n$  with $\bmx_{i} \sim N_d(\bm0, \bm\Sigma_x)$ and $\bmu_{i} \sim N_d(\bm0, \bm\Sigma_u)$. Here, $\bmw_i = (w_i^1,\ldots, w_i^{d})$, with the subscript and superscript denoting the observation and components respectively. Denote $\bm{W}$ as the $n \times d$ matrix of observed data. The model is equivalent to the following hierarchical representation. First, the latent random variables $\bmx_i$ are generated from a $N_d(\bm{0}, \bm{\Sigma}_x)$ distribution, and when conditioned on $\bmx_i$ and $\bm\Sigma_u$, we have  $\bmw_i\vert \bmx_i, \bm\Sigma_u \sim N_d(\bmx_i,\bm{\Sigma}_u)$ for each $i = 1, \ldots, n$.  This forms an intuitive generative process, where first $\bmx$ is realized, then contaminated by measurement error $\bmu$, and observed finally as $\bmw$. The problem of interest is to estimate the precision matrix $\bm\Omega_x$ in the high dimensional setting $n<d$.

	When no measurement error is present, i.e the $\bmx_i$ are directly observed, the sample covariance matrix $\bm{S} = n^{-1} \sum_{i=1}^n (\bmx_i - \bar{\bmx}) (\bmx_i - \bar{\bmx})^\top$, with $\bar{\bmx}$ being the sample mean, is a consistent estimator for $\bm{\Sigma}_x$. However it has the rank of at most $n<d$, so it is not invertible to estimate $\bm{\Omega}_x$. When measurement error is present, we assume the covariance matrix of measurement error $\bm\Sigma_u$ is known or estimable from replicates. A naive approach is first to estimate $\bm\Sigma_x$ by $ \tilde{\bm\Sigma}_x=\bm{S}_w - \bm\Sigma_u$, where $\bm{S}_w$ denotes the sample covariance from contaminated data $\bm{W}$, and then to invert $\tilde{\bm\Sigma}_x$ to estimate $\bm{\Omega}_x$. The main issue with this approach is that $\tilde{\bm{\Sigma}}_x$ is generally not positively definite.  This implies its inverse is also not positively definite, which is necessary to find a consistent estimate $\bm{\Omega}_x$.	Hence, a correction procedure to estimate $\bm{\Omega_x}$ need not rely upon the sample covariance matrix $\tilde{\bm\Sigma}_x$ directly. Furthermore, the procedure is also able to incorporate sparsity constraints to recover the graphical model structure. These requirements are addressed by the procedure described in the next section. 
\section{The IRO-BAGUS Algorithm} \label{section:ImputationRegularizationME}

In a recent work, \cite{liang2018imputation} develop a methodology to efficiently handle high dimension problems with missing data.  Their solution is an EM-algorithm variant which alternates between two steps, the imputation step and regularized optimization step; we refer to their algorithm as the IRO algorithm.  Denote the missing data as $Y$, and observed data as $X$. Also denote the desired parameter to be estimated by $\theta$, and begin with some initial guess $\theta^{(0)}$.  During the $t^{\text{th}}$ iteration, the IRO algorithm generates Y from the distribution given by the current estimate of $\theta$, i.e. $Y \sim \pi(Y \vert X , \theta^{(t - 1)})$.  Then, using $X$ and $Y$, maximizes $\theta$, under regulation, using the full likelihood.  \cite{liang2018imputation} show that this procedure results in a consistent estimate of $\theta^{(t)}$, and results in a Markov chain with stationary distribution.

We make use of this framework for our current problem pertaining to mismeasured observations instead of missing values.  The problems are naturally related in the sense that both are generating values of the true process from some estimated underlying distribution.  We return to the hierarchical structure of the problem, i.e. $\bmw \sim N_d(\bmx , \bm{\Sigma_u})$ and $\bmx \sim N_{d}(\bm0 , \bm{\Sigma_x})$. The IRO algorithm proceeds iteratively between the two following steps:
\begin{itemize}
	\item \textit{Imputation step}: At iteration $t$, draw $\bm{X}^{(t)} = (\bmx_1^{(t)}, \ldots,\bmx_d^{(t)})$ from the posterior full conditional of $\bmX$, using the current estimate of $\bm{\Omega_x}^{(t-1)}$. Specifically, for $i=1,\ldots,n$, draw $\bmx_i^{(t)} \vert \bmw , \bm{\Omega_x^{(t-1)}} \sim N_d(\bm{\Lambda^{-1}} \bm{\Omega_u} \bmw_i , \bm{\Lambda}^{-1})$ where $\bm{\Lambda} = (\bm{\Omega_x}^{(t-1)} + \bm{\Omega_u})$. Note that the posterior distribution of $\bmx_i$ depends only on $\bmw_i$ due to independence. This allows for easy generation of data from the true underlying distribution.  
	
	\item \textit{Regularization Step}: Apply a regularization to the generated $\bmX^{(t)}$ and obtain a new MAP estimate of $\bm{\Omega_x}^{(t)}$. 
\end{itemize}

In this work, the regularization step is carried out based on a recent Bayesian methodology, called BAGUS. Hence, the whole algorithm is referred to as the IRO-BAGUS algorithm. The next subsections \ref{SSLPriorGGM}-\ref{section:BagusVariableSelection} outline prior specifications, the full model, and variable selection for BAGUS. After that, section \ref{section: consistency} discusses consistency of the IRO-BAGUS estimate.

	
\subsection{The Spike-and-Slab Lasso Prior Specification} \label{SSLPriorGGM}

Denote the elements $\bm{\Omega}_x$ to be $\omega_{ij}$.  Recently, a non-convex, continuous relaxation penalty for the spike-and-slab prior was created for the standard lasso problem \cite{rovckova2018spike}.  This prior was extended to the case of graphical models by \cite{gan2018bayesian}, and is given by 
\begin{equation}
\pi(\omega_{ij}) = \frac{\eta}{2 v_1} \exp\left\{- \frac{\vert \omega_{ij} \vert}{v_1}\right\} + \frac{1 - \eta}{2 v_0} \exp\left\{- \frac{\vert \omega_{ij} \vert }{v_0}\right\}
\end{equation}
for the off diagonal elements ($i \neq j$), where $0 < v_0 < v_1$ and $0<\eta<1$.  This prior can be interpreted as a mixture of the spike-and-slab prior. The first component of the mixture has prior probability $\eta$, and is associated with the slab component, i.e. $\omega_{ij} \neq 0$.  Conversely, with prior probability $1 - \eta$ the element is from the spike component, suggesting $\omega_{ij} = 0$. 

Traditionally, the spike-and-slab prior has a point mass component at 0 and some other continuous distribution for the slab component.  This is to represent setting unwanted terms exactly to 0.  Here, both the spike and the slab components are distributed according to a Laplace distribution; both are centered at 0, but the spike is more tightly centered by a smaller variance term than the slab.  This relaxation of the spike-and-slab prior allows for efficient gradient based algorithms, while still being theoretically sound as shown in \cite{rovckova2018bayesian}.

Shrinkage is not desired on the diagonal elements, so a weakly informative exponential prior is given instead, $ \pi(\omega_{ii}) = \tau \exp\left\{ - \tau \omega_{ii} \right\}.$  Another consideration for the prior of $\bm\Omega_x$ is to ensure the whole matrix to be positive definite, denoting as $\bm\Omega_x \succ 0$.  Moreover, in line with \cite{gan2018bayesian}, we require the spectral norm to be bounded above by some value $B$, $\vert \vert \bm\Omega_x \vert \vert \leq B$.  This assumption will be important going forward.  The full prior distribution for $\bm\Omega_x$ is then given by 
\begin{equation}
\pi(\bm\Omega_x) = \prod_{i < j} \pi(\omega_{ij}) \prod_{i} \pi(\omega_{ii}) I(\bm\Omega_x \succ 0) I(\vert \vert \bm\Omega_x \vert \vert \leq B).
\end{equation}

\subsection{The Full Model}
Without measurement error, the posterior distribution is specified as 
\begin{equation}
\pi(\bm\Omega_x \vert \bmX) \propto \prod_{i = 1}^{n} \pi(\bmx_i \vert \bm\Omega_x) \pi(\bm\Omega_x).
\label{eq:truePosterior}
\end{equation}
The full conditionals can be derived for \eqref{eq:truePosterior}, but, to avoid costly MCMC sampling for this large dimensional problem, \cite{gan2018bayesian} opted to instead find the mode of the the posterior distribution, often referred to as the MAP.  The MAP can be found by minimizing the uncontaminated (UC) objective
\begin{equation} 
L^{\text{UC}}(\bm\Omega_x)  =  \log \pi(\bm\Omega_x \vert \bmX) 
=  \frac{n}{2}\left( \text{tr}(\bm{X}^T \bm\Omega_x \bm{X}) - \text{logdet}(\bm\Omega_x)\right) + \sum_{i < j} \pi(\omega_{ij}) + \sum_{i} \pi(\omega_{ii}) + K
\label{eq:trueLogPosterior}
\end{equation}
with respect to $\bm\Omega_x$, where $K$ is the normalizing constant in \eqref{eq:truePosterior}. To this end, \cite{gan2018bayesian} proved the local convexity of \eqref{eq:truePosterior} when $\vert \vert \bm\Omega_x \vert \vert \leq B < \infty$, which allows an easy optimization procedure that converges asymptotically to the correct precision matrix.

\subsection{Variable Selection} \label{section:BagusVariableSelection}

Many practictioners use Gaussian graphical models for the purpose of identifying non-zero entries of $\bm\Omega_x$, which signify conditional dependencies among the two different variables.  The spike-and-slab lasso formulation allows for this quite easily by viewing the optimization as an instance of the EM-algorithm and defining the hierarchical prior
\begin{equation}
\begin{cases}
\omega_{ij} \vert r_{ij} = 0 \sim \text{Laplace}(0 , v_0) \\
\omega_{ij} \vert r_{ij} = 1 \sim \text{Laplace}(0 , v_1)
\end{cases}.
\end{equation}
Here, $r_{ij}$ is the random indicator that the element of the precision matrix follows from the spike or the slab component, where $r_{ij} \sim \text{Bern}(\eta)$.
A further hierarchical level can be added by treating $\eta$ as random instead of a fixed hyperparameter.  Recent work from \cite{deshpande2017simultaneous} illustrates this and is line with the spike-and-slab Lasso setting of \cite{rovckova2018spike}.  Given our purpose is to study the effect of the measurement error, we choose to treat it as a fixed.

The conditional posterior distribution for $r_{ij}$ is also Bernoulli, with probability of success 
\begin{equation}
p_{ij} = \frac{v_1}{v_0}\frac{1 - \eta}{\eta} \exp\left\{ \vert \omega_{ij} \vert \left(\frac{1}{v_1} - \frac{1}{v_0} \right) \right\}.
\label{eq:E_prob}
\end{equation}
We will use the MAP estimate of $\omega_{ij}$ in \eqref{eq:E_prob} as the approximate probability of inclusion.  With the inclusion probability a hard threshold will be specified to determine the final inclusion for the purpose of simulation and model selection.  Denote $\bm{R}$ and $\bm{P}$ to be the matrix of indicators and conditional posterior probability of inclusion for each element of $\bm\Omega_x$. We note that for final inference it may be better to forego this inclusion threshold, and instead rank-order the $p_{ij}$ for purposes of downstream investigation; however, this will depend on the application at hand.

\subsection{Consistency of the IRO-BAGUS algorithm}
\label{section: consistency}
The entire data generation process for the contaminated sample is summarized below:
\begin{align*}
&\bmw_i \vert \bmx_i, \bm\Omega_x \sim N_d(\bmx_i, \bm\Sigma_u),~i=1,\ldots,n \\
&\bmx_i \vert \bm{\Omega}_x \sim N_d(\bm{0},\bm{\Omega}_x^{-1}),~i=1,\ldots,n \\
&\omega_{ij} \vert r_{ij} = 0,v_0 \sim \text{Laplace}(0,v_0),~i\neq j, ~i,j=1,\ldots,n \\
&\omega_{ij}\vert r_{ij} = 1,v_1 \sim \text{Laplace}(0,v_1),~i\neq j, ~i,j=1,\ldots,n \\
&\omega_{ii} \sim \text{Exp}(\tau),~i=1,\ldots,n \\
&r_{ij}\vert \eta \sim \text{Bern}(\eta),~i\neq j,! i,j=1,\ldots,n. 
\end{align*}
Instead of approximating the posterior distribution of all the parameters, the IRO-BAGUS algorithm iteratively generates realizations of uncontaminated data, $\bmX$, then optimizes $\bm\Omega_x$ with these generated values.   
Under some technical conditions, the IRO algorithm is shown to produce a consistent estimate after each iteration in the context of missing data when the regularization step results in a consistent estimate \cite{liang2018imputation}. We show that these conditions are also held in the case of contaminated data, so the IRO-BAGUS algorithm results in a consistent estimate.  Theorem \ref{thm:iro_consistent} is the analogue statement of consistency as in the missing data case.  The proof is given in the appendix.
\begin{theorem}
	Assuming $\vert \vert \bm\Omega_x \vert \vert \leq B$, then the estimate $\bm\Omega_x^{(t)}$ is uniformly consistent to $\bm\Omega_x$ when $\log(t) = \mathcal{O}(n)$. 
	\label{thm:iro_consistent}
\end{theorem}

It can be seen that the nature of the IRO algorithm is similar to that of MCMC.  Under a few more conditions, \cite{liang2018imputation} note that the IRO results in a Markov chain with a stationary distribution, and hence the average of the maximization steps are consistent estimates of the underlying parameters.  Our final estimates are the averaged regularized optimization steps given by BAGUS from the imputed data at each iteration, removing a small number of the beginning iterations as burn-in.  By averaging instead of taking only the final iteration, we make the analysis less variable. In this sense, the relationships that the correction procedure identifies are more likely to be true relationships, cutting down on the number of false positives.

\section{Computation for the IRO-BAGUS algorithm}
\subsection{Finding MAP estimate for $\bm\Omega_x$} \label{section:BagusIRO}

Here we consider some computational aspects of our proposed methodology.  
First, we focus to the optimization of $\bm\Omega_x$. In our procedure, once $\bmX$ is generated, the objective function to be optimized is $L^{\text{uc}}$, as was shown in Equation \eqref{eq:E_prob}; we note this is due to the conditional independence of $\bmW$ and $\bm\Omega_x$ in the hierarchical structure of the contamination process.  Optimizing $L^{\text{uc}}$ is difficult to do directly; therefore, the latent factors $r_{ij}$ from Section \ref{section:BagusVariableSelection} are introduced into the process as in \cite{gan2018bayesian}.  This allows an E-step similar to the spike-and-slab Lasso and an M-step similar to the Graphical Lasso.

The optimization seeks to find the MAP of the posterior proportional to
\begin{equation}
\vert \bm\Omega_x \vert^{\frac{1}{2}}\exp\left\{-\frac{1}{2}\bmX^T \bm\Omega_x \bmX \right\} 
\prod_{i < j} \pi(\omega_{ij} \vert r_{ij}) \pi(r_{ij} \vert \eta)
\prod_{i} \pi(\omega_{ii} \vert \tau) 
I(\bm\Omega_x \succ 0) I(\vert \vert \bm\Omega_x \vert \vert \leq B).
\end{equation}
The E-step takes the conditional expectation of $r_{ij}$ in the proportional posterior.  Each $r_{ij}$ is conditionally Bernoulli with probability as given in Equation \eqref{eq:E_prob}, which allows for easy calculation of the desired conditional expectations. Then, the desired $Q$ function to maximize in the M-step is given by 
\begin{equation}
Q(\bm\Omega_x \vert \bm\Omega_x^{(t)}) = \mathbb{E}_{\bm{R} \vert \bm\Omega_x^{(t)}}\log\pi(\bm\Omega_x , \bmX \vert \bmW , \bm\Sigma_u),
\end{equation}
where the expectation is taken element wise for $\bm{R}$ by assumed independence of inclusion.  Maximizing $Q$ is done by a block coordinate descent algorithm.  The algorithm cycles between column-wise updates of $\Omega_x$.  We put the details of this procedure in the Appendix.

\subsection{Other Computation Considerations} \label{section:BagusIROCompRemarks}


\subsubsection{Estimating $\bm\Sigma_u$}
We have assumed the covariance matrix of measurement error $\bm\Sigma_u$ to be known before applying the IRO-BAGUS algorithm.  In practice, the matrix $\bm\Sigma_u$ is often estimated from ancillary data, such as replicate observations. Assuming measurement error between variables to be independent is reasonable for many problems and often used in the literature \cite{sorensen2015measurement}. In that case, only the diagonal of $\bm\Sigma_u$ only needs to be estimated. For the data analysis application we provide in Section \ref{section:MicroarrayData}, we estimated them with the method described in \cite{turro2007bgx}, assuming homogeneity of measurement error between observations. After that, we performed the IRO-BAGUS algorithm as previously described.
\subsubsection{Starting Values}
The starting value plays a significant role in the speed of optimization at each step.  To begin, we perform a naive analysis on the raw contaminated data, $\bmW$, giving estimate $\bm\Omega_x^{(0)}$.  This value is then used to start the IRO procedure by generating $\bmX$.  Each optimization has a warm start from the previous iteration's estimated precision matrix.

\subsubsection{Addressing the constraint, $\vert \vert \bm\Omega_x \vert \vert \leq B$}
The constaint that $\vert \vert \bm\Omega_x \vert \vert \leq B$ needs to be addressed.  \cite{gan2018bayesian} suggest using a threshold on the largest absolute value of the elements of the column being updated in the block coordinate descent.  We use the same threshold, and find no performance issues when used with the IRO algorithm.  

\subsubsection{Positive-Definiteness of $\bm\Omega_x$}
Many procedures to estimate a sparse precision can not guarantee postive-definiteness, however \cite{gan2018bayesian} show that the output of BAGUS from the EM algorithm is always symmetric and positive definite.  It is easy to show that the imputation step, with final results averaged, also results in this nice property.
\begin{theorem}
The estimated precision matrix $\hat{\bm\Omega}_x = T^{-1} \sum_{t = 1}^{T} \bm\Omega_x^{(t)}$ is symmetric and positive definite if the initial value of $\bm\Omega_x$ for BAGUS at each iteration $t$ was also positive definite.
\end{theorem}
\begin{proof}
By Theorem 5 in \cite{gan2018bayesian}, if the initial value to optimize BAGUS is positive definite, then $\bm\Omega_x^{(t)}$ is also positive definite.  The set of positive definite matrices form a cone, and hence the average will also be in this cone.  
\end{proof}

\subsubsection{Parameter Tuning} 

There are four hyperparameters in BAGUS, $\eta , \tau, v_0, \text{and } v_1$.  As with \cite{gan2018bayesian}, we always set $\eta = 0.5$ and $\tau = v_0$, which leaves two hyperparameters to tune.  Again, we follow \cite{gan2018bayesian}, who suggest a BIC-like criteria to select the best model from a grid of hyperparameters.  This criteria is 
\[
\text{BIC} = n(\text{tr}(\bm{S} \hat{\bm\Omega}_x) - \text{logdet}(\hat{\bm\Omega}_x)) + \log(n) \times q,
\]
where $\hat{\bm\Omega}_x$ is the estimated precision matrix and $q$ is the number of non-zero elements of the estimated in the upper diagonal of the precision matrix. We use this in similar fashion for the IRO procedure, but instead we use the averaged $\bm\Omega_x^{(t)}$ in the BIC calculation.

\section{Simulation}
\subsection{Simulation Setup} \label{section:BagusIROSimSetup}
We investigate the performance of our methodology under several different settings.  For each setting we generate  $\bmx_i$ following a $d$-variate Gaussian distribution with mean $\bm{0}$ and precision matrix $\bm{\Omega}_x$ according to some graph structure; we refer to this as the \textit{true data}.  Then, the contaminated observation $\bmw_i$ was generated from $\bmw_i = \bmx_i+\bmu_i$, where $\bmu_i \sim N_d(\bm0, \bm\Sigma_u),~i=1,\ldots,n$. The measurement error covariance matrix $\bm{\Sigma}_u$ is assumed to be a diagonal matrix, with element $\left[\bm{\Sigma}_u\right]_{ii} = \gamma\left[\bm{\Sigma}_x\right]_{ii}$, where $\left[\bm{\Sigma}_x\right]_{ii}$ is the variance of the dimension $x^i$. In other words, the constant $\gamma$ controls the noise-to-signal ratio on each variable. For the purposes of simulation, we assume the amount of measurement error to be known.

To generate the true data we use the \textit{huge} package \cite{zhao2012huge}.  We inspect two different types of graphs, referred to as \textit{hub} and \textit{random}; we expand on these below where $\omega_{ij}$ denotes the $(i,j)$ element of $\bm\Omega_x$.
\begin{enumerate}
\item Hub: For $d / 20$ groups, $\omega_{ij} = \omega_{ji} = 1$ if in the same group. $\omega_{ij} = 0$ otherwise.
\item Random: For $1 \leq i < j \leq d$, $\omega_{ij} = 1$ with probability $\frac{3}{d}$, 0 otherwise.
\end{enumerate}
We illustrate the stuctures in Figure \ref{figure:structure}.

\begin{figure}
\centering
\begin{tabular}{c c}
\includegraphics[page = 2 , scale = .4]{Structure_Graphs.pdf}
\includegraphics[page = 1 , scale = .4]{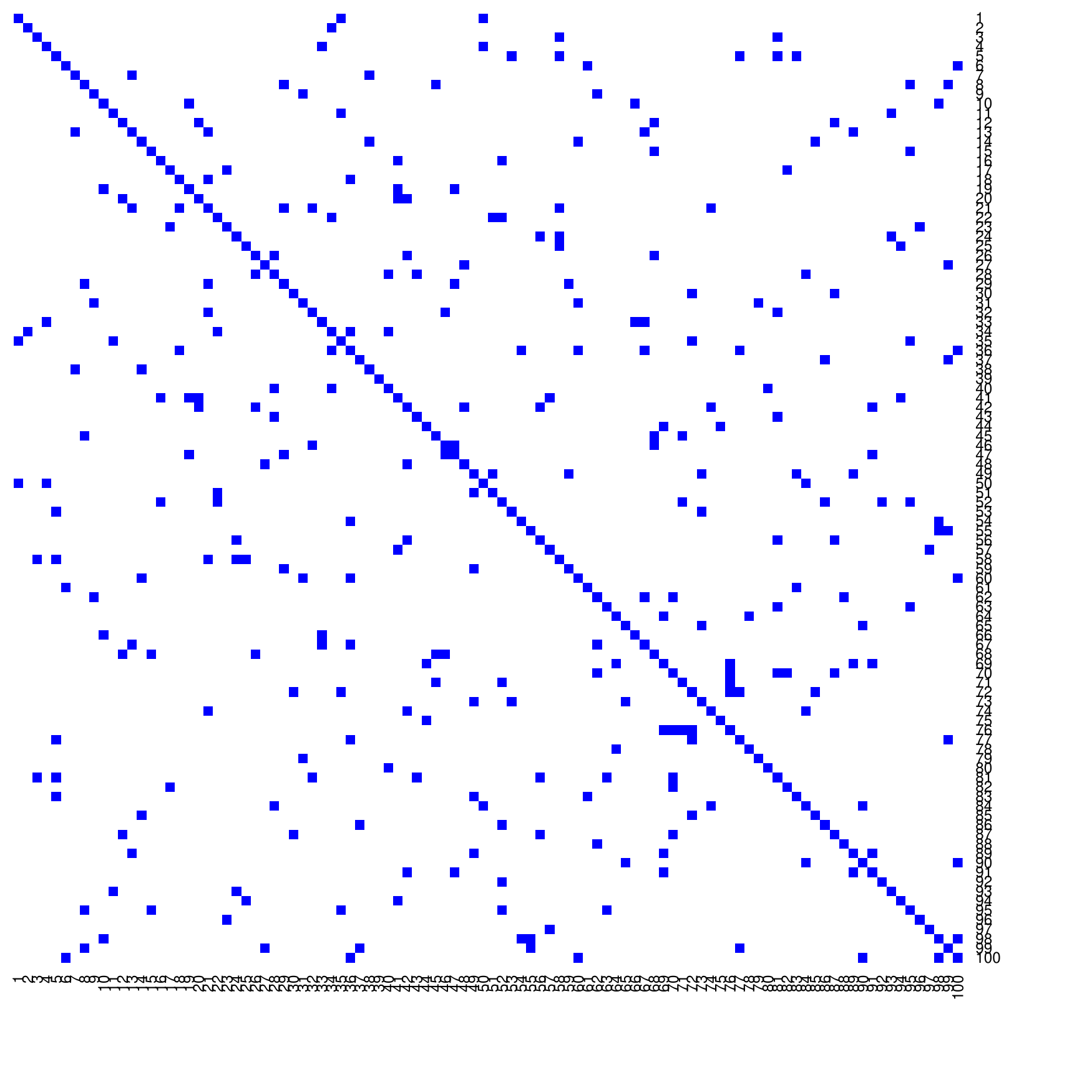}
\end{tabular}
\caption{Graphical representation for $d = 100$ of the hub (left) and random (right) structure, respectively. Note that the random graph is subject to change due to the randomness.}
\label{figure:structure}
\end{figure}

Each model was generated with $n = 100$ observations.  We inspect each model for $d = \{100 , 200\}$ and $\gamma = \{0.1 , 0.25, 0.5\}$.  The amount of correction-imputations was set to be 50, with the first 20\% discarded as burn-in; we note that we inspected 25 and 100 imputations with the same percentage of burn-in samples with minimal differences in output.  Each setting was replicated 50 times, and the final results are the average of these replicates.  Hyperparameter tuning was done as described in Section \ref{section:BagusIROCompRemarks}.  Because measurement error is often ignored in the context of GGMs, our simulations also provide perspective onto the negative effect that measurement error can impose into the model performance.  

To inspect model performance, we examine both the estimated precision matrix and the ability to do variable selection of BAGUS on the true data (true), BAGUS on the contaminated data (naive), and our IRO-BAGUS methodology on the contaminated data (corrected). For each estimated precision matrix $\hat{\bm{\Omega}}_x$,  estimation error is measured by $\vert\vert \hat{\bm{\Omega}}_x-\bm{\Omega}_x\vert\vert_F$, and variable selection is evaluated by different metrics involving the true positives (TP), false positives (FP), true negatives (TN), and false negatives (FN) are reported: specificity (SPE), sensitivity (SEN), precision (PRE), accuracy (ACC), and Matthews correlation coefficient (MCC); these values are defined as
\[\begin{split}
& \text{SPE} = \frac{\text{TN}}{\text{TN + FP}}, \qquad \text{SEN} = \frac{\text{TP}}{\text{TP + FN}}, \\
& \text{PRE} = \frac{\text{TP}}{\text{TP + FP}}, \qquad \text{ACC} = \frac{\text{TP + TN}}{\text{TP + FP + TN + FN}} \\
& \text{MCC} = \frac{\text{TP} \times \text{TN} - \text{FP} \times \text{FN}}{\text{(TP + FP)(TP + FN)(TN + FP)(TN + FN)}}.
\end{split}\]
Additionally, we also report the area under the ROC curve (AUC), which gives insight into the amount of seperation of the classification.  These different metrics give insight into the tradeoffs and gains of each setting.

\subsection{Simulation Results}

Table \ref{table:hub} and Table \ref{table:random} present the results for the hub and random structure, respectively.  To begin, we note the effect of the increasing measurement error.  This can be observed by examining the growing difference in the performance of the true and naive model when holding $d$ fixed and increasing the amount of contamination.  Focusing on the hub structure, a decrease in the quality of selection and estimation can be observed for each setting, which grows worse with more contamination.  The selection accuracy metrics with respect to the prespecified 0.5 cut-off show drops in performance of around $50\%$.  The estimated precision matrix from the naive grows worse with measurement error, and is also about $50\%$ worse when the signal-to-noise is 0.5.

\begin{table}
\centering
\begin{tabular}{cc  lrrrrrrr}
  \hline
$\gamma$ & d & Model & SEN & SPE & PRE & ACC & MCC & FROB & AUC \\ 
  \hline
\multirow{6}{*}{0.1} & \multirow{3}{*}{100} & True & 1.00 & 0.65 & 0.85 & 0.99 & 0.73 & 5.11 & 0.95 \\ 
   && Naive & 1.00 & 0.50 & 0.76 & 0.99 & 0.61 & 6.81 & 0.93 \\ 
   && Corrected & 1.00 & 0.51 & 0.78 & 0.99 & 0.62 & 6.17 & 0.97 \\ 
\cline{2-10}
& \multirow{3}{*}{200}  & True & 1.00 & 0.67 & 0.77 & 0.99 & 0.71 & 7.36 & 0.94 \\ 
   && Naive & 1.00 & 0.51 & 0.69 & 0.99 & 0.59 & 9.63 & 0.92 \\ 
   && Corrected & 1.00 & 0.51 & 0.71 & 0.99 & 0.60 & 8.68 & 0.96 \\ 
\hline \hline
\multirow{6}{*}{0.25} & \multirow{3}{*}{100}  & True & 1.00 & 0.66 & 0.84 & 0.99 & 0.74 & 5.09 & 0.95 \\ 
  && Naive & 0.99 & 0.38 & 0.60 & 0.98 & 0.47 & 8.54 & 0.90 \\ 
   && Corrected & 1.00 & 0.36 & 0.68 & 0.98 & 0.49 & 7.71 & 0.94 \\ 
\cline{2-10}
 & \multirow{3}{*}{200}  & True & 1.00 & 0.67 & 0.78 & 1.00 & 0.72 & 7.29 & 0.94 \\ 
   && Naive & 1.00 & 0.40 & 0.52 & 0.99 & 0.45 & 12.10 & 0.90 \\ 
   && Corrected & 1.00 & 0.37 & 0.62 & 0.99 & 0.48 & 10.68 & 0.95 \\ 
\hline \hline
 \multirow{6}{*}{0.5} & \multirow{3}{*}{100}  & True & 1.00 & 0.66 & 0.85 & 0.99 & 0.74 & 5.03 & 0.95 \\ 
   && Naive & 1.00 & 0.20 & 0.50 & 0.98 & 0.31 & 9.67 & 0.84 \\ 
   && Corrected & 1.00 & 0.20 & 0.70 & 0.98 & 0.37 & 8.74 & 0.89 \\ 
\cline{2-10}
 & \multirow{3}{*}{200}  & True & 1.00 & 0.68 & 0.77 & 0.99 & 0.72 & 7.36 & 0.94 \\ 
   && Naive & 1.00 & 0.20 & 0.37 & 0.99 & 0.27 & 13.70 & 0.83 \\ 
   && Corrected & 1.00 & 0.17 & 0.59 & 0.99 & 0.31 & 12.53 & 0.89 \\ 
   \hline
\end{tabular}
\caption{Simulation results for the hub graph structure, as specified in Section \ref{section:BagusIROSimSetup}.  For each signal-to-noise ratio and $d$, the true, naive, and corrected models are shown for metrics defined in Section \ref{section:BagusIROSimSetup}.}
\label{table:hub}
\end{table}

We now turn attention to the performance of the correction step.  First, take note of the first five metrics which are  based on the confusion matrix for the 0.5 cutoff threshold.  Averaging across the IRO iterations was expected to result in an analysis that favored identifying relationships that were more certain, which can be observed by inspection of the precision (PRE).  The gains from the precision are most notable as $d$ grows larger, and more pair-wise relationships exist; when $d = 200$, we note nearly 10\% and 50\% performance gains in the precision for signal-to-noise ratios of 0.25 and 0.5, respectively.  In both the hub and random structure the naive and corrected models perform similarly in terms of the sensitivity, specificity, accuracy, and MCC.  

\begin{table}
\centering
\begin{tabular}{cc lrrrrrrr}
  \hline
 Amt. ME & d & Model & SEN & SPE & PRE & ACC & MCC & FROB & AUC \\ 
  \hline
\multirow{6}{*}{0.1} & \multirow{3}{*}{100} & True & 1.00 & 0.42 & 0.84 & 0.98 & 0.59 & 4.61 & 0.89 \\ 
  & & Naive & 1.00 & 0.32 & 0.80 & 0.98 & 0.50 & 5.34 & 0.88 \\ 
  & & Corrected & 1.00 & 0.32 & 0.80 & 0.98 & 0.50 & 5.04 & 0.91 \\ 
\cline{2-10}
& \multirow{3}{*}{200} & True & 1.00 & 0.36 & 0.76 & 0.99 & 0.52 & 6.72 & 0.86 \\ 
  & & Naive & 1.00 & 0.30 & 0.66 & 0.99 & 0.44 & 7.61 & 0.85 \\ 
  & & Corrected & 1.00 & 0.28 & 0.68 & 0.99 & 0.43 & 7.25 & 0.90 \\ 
\hline \hline
\multirow{6}{*}{0.25} & \multirow{3}{*}{100} & True & 1.00 & 0.45 & 0.86 & 0.98 & 0.61 & 4.66 & 0.90 \\ 
  & & Naive & 1.00 & 0.27 & 0.70 & 0.97 & 0.42 & 6.68 & 0.85 \\ 
  & & Corrected & 1.00 & 0.23 & 0.75 & 0.97 & 0.40 & 5.72 & 0.88 \\ 
\cline{2-10}
& \multirow{3}{*}{200} & True & 1.00 & 0.37 & 0.75 & 0.99 & 0.52 & 6.77 & 0.86 \\ 
  & & Naive & 1.00 & 0.23 & 0.52 & 0.99 & 0.34 & 9.04 & 0.82 \\ 
  & & Corrected & 1.00 & 0.18 & 0.60 & 0.99 & 0.32 & 7.95 & 0.86 \\ 
\hline \hline
\multirow{6}{*}{0.5} & \multirow{3}{*}{100} & True & 1.00 & 0.43 & 0.85 & 0.98 & 0.59 & 4.65 & 0.89 \\ 
  & & Naive & 1.00 & 0.14 & 0.55 & 0.97 & 0.26 & 7.71 & 0.79 \\ 
  & & Corrected & 1.00 & 0.09 & 0.67 & 0.97 & 0.24 & 6.49 & 0.79 \\ 
\cline{2-10}
& \multirow{3}{*}{200}& True & 1.00 & 0.37 & 0.76 & 0.99 & 0.53 & 6.74 & 0.86 \\ 
 & & Naive & 1.00 & 0.12 & 0.39 & 0.98 & 0.21 & 10.42 & 0.77 \\ 
  & & Corrected & 1.00 & 0.06 & 0.56 & 0.99 & 0.18 & 8.92 & 0.78 \\ 
   \hline
\end{tabular}
\caption{Simulation results for the random graph structure, as specified in Section \ref{section:BagusIROSimSetup}.  For each signal-to-noise ratio and $d$, the true, naive, and corrected models are shown for metrics defined in Section \ref{section:BagusIROSimSetup}.}
\label{table:random}
\end{table}

It seems at first glance that the selection performance, ignoring the precision, of the correction procedure is comparable to the naive, but these discrepencies can be attributed to the prespecified inclusion cut-off on the $\bm{P}$ matrix.  In practice it can often be more reasonable to rank order the inclusion probabilities to identify relationships to further investigate in future experiments.  With this in mind, we turn to the performance with respect to the AUC where consistent improvements can be seen for the hub and random structure in most all settings.   The AUC helps understand the amount of seperation found in the model across all thresholds, which helps justify that the correction step is making improvements in seperating the classes for the true relationships as AUC improvements are seen in all but the random graph with $d = 200$ and signal-to-noise ratio of $0.5$.  

We note two items in regard to the AUC.  First, the AUC of the corrected model sometimes outperforms the true model, too. In particular, this happens in the hub structure when the amount of measurement error is $0.1$. This can be attributed to the measurement error in models that are easily identified. Second, in the random structure with the amount of measurement error being $0.5$, the corrected model does not make substantial improvements in results over the naive model.  We note the difficulty of this setting, as the random structure often performs worse than other structures in identification, and now we add more noise via the contamination.  With a relatively small sample size, this noise is difficult to overcome.

Finally, we note the quality of the estimated precision matrix, as measured by Frobenious norm of the difference.  In every setting for both the hub and random matrices, the corrected model outperforms the naive model's estimate of the precision matrix.  In the hub structure this improvement is often of the order of $15$-$20\%$ better, while in the random structure a $10\%$ improvement is generally observed.  If the intent of the analysis is to use the estimated precision matrix in downstream analysis, this can result in more refined results.

\section{Data Analysis} \label{section:MicroarrayData}

A common source of noise in analysis involving gene expression datasets is measurement error \cite{rocke2001model}.  Gaussian graphical models are often used to inspect the relationship of different genes in varying experiments \cite{kramer2009regularized}.  We illustrate our methodology using an Affymetrix microarray dataset containing 144 subjects of favorable histology Wilms tumors hybridized to the Affymetrix Human Genome U133A Array \cite{huang2009predicting}. The data is publicly available on the GEO website, dataset GSE10320 uploaded 1/30/2009.  A feature of Affymetrix data, and many other gene expression measurement platforms, is the use of multiple probes for each gene for each patient, giving replicate measurements for each patient's gene measurement.  The replicates for each patient enable an estimate of the measurement error, where we again assume the amount of contamination is independent across genes.  

We follow the preprocessing steps taken in \cite{sorensen2015measurement} and \cite{nghiem2018simulation}, which used this study in the context of measurement error in variable selection for linear models.  The process begins by processing the raw data with the Bayesian Gene Expression (BGX) package \cite{turro2007bgx}.  BGX creates a posterior distribution for the log-scale expression level of each gene in each sample.  The study recorded measurements for 22283 different genes.  

To remove unnecessary computational burden, we reduced the number of genes by applying four different filters in the following order.  The first filter removes expression values that do not have a corresponding Entrez gene ID in the NCBI database \cite{o2015reference}.  The second filter removes expression values with low variability by requiring at least 25\% of samples to have intensities above 100 fluorescence units.  The third filter removes expression values with low variability by requiring the interquartile range to be at least 0.6 on the log scale.  The last filter removes expression values that have have an error to signal to noise ratio greater than 0.5, which we discuss in more depth below.  After filtering, there were 273 expression values remaining for the analysis.  

Now, we discuss how we estimate the measurement error of each gene.  We assume that the measurement error variance is constant across patients for a given gene.  We also assume that the measurement error is independent for each gene, and need not be equal for each gene. Let $\bm{\hat{\mu}} = (\hat{\mu}_{1j} , \ldots , \hat{\mu}_{nj})^T$ denote the estimated vector of the patient's gene expression levels for gene $j$.  Further, let $\bar{\mu} = n^{-1} \sum_{j = 1}^n \hat{\mu}_{ij}$ and $\hat{\sigma}_{j}^2 = n^{-1} \sum_{j = 1} (\hat{\mu}_{ij} - \bar{\mu}_j)^2$ denote the mean and variance of each gene, respectively.  For patient $i$, standardized measurements are given by $\bm{W}_i = (W_{i1} , \ldots , W_{ip})$, calculated as $W_{ij} = \hat{\sigma}_j^{-1}(\hat{\mu}_{ij} - \bar{\mu}_{j})$ for each $j = 1 , \ldots , 273$.

Let $\text{var}(\hat{\mu}_{ij})$ denote the posterior variance of the estimated distribution of patient $i$'s gene $j$.  These estimates are then combined as $\hat{\sigma}_{u,j}^2 = n^{-1} \sum_{i = 1}^n \text{var}(\hat{\mu}_{ij})$.  The measurement error covariance matrix of the standardized data $\bm{W}$ is then estimated by diagonal matrix $\hat{\bm{\Sigma}}_u$, where $(\hat{\bm{\Sigma}}_u)_{j,j} = \hat{\sigma}_{u,j}^2 / \hat{\sigma}_j^2$ for $j = 1 , \ldots , p$ and off-diagonal elements are 0.  The fourth filter can be now formalized, where genes are removed if $\hat{\sigma}_{u , j}^2 \geq 0.5 \hat{\sigma}_j^2$; i.e. only genes with a noise-to-signal ratio less than 1 are kept for the analysis.

The original BAGUS algorithm and the corrected BAGUS algorithm were run for the remaining genes found after filtering.  As with the simulations, the corrected BAGUS found fewer conditional pair-wise relationships; for this data set, the original BAGUS and corrected BAGUS found 1045 and 552 conditional pair-wise relationships, respectively.  Of the 1045 naive pair-wise relationships, 42\% were also found in the corrected pair-wise relationships; similarly, of the 552 corrected conditional pair-wise relationships, 80\% were found in the naive model.  The large percentage overlap of relationships in the corrected model with relationships in the naive model suggests that most relationships in the corrected model are true relationships. Conversely, the small percentage overlap of relationships in the naive model with those in the correct model suggests that the naive model is finding many false positive relationships. We illustrate the conditional pair-wise dependencies of the genes in Figure \ref{figure:NetworkGraph}.  The naive analysis is shown on the left and the corrected on the right, where the green edges signify relationships found by both procedures and purple edges signify procedure specific relationships.

\begin{figure}
\begin{tabular}{c c}
\includegraphics[page = 1 , scale = .4 , trim = {5cm 4cm 3cm 4cm}]{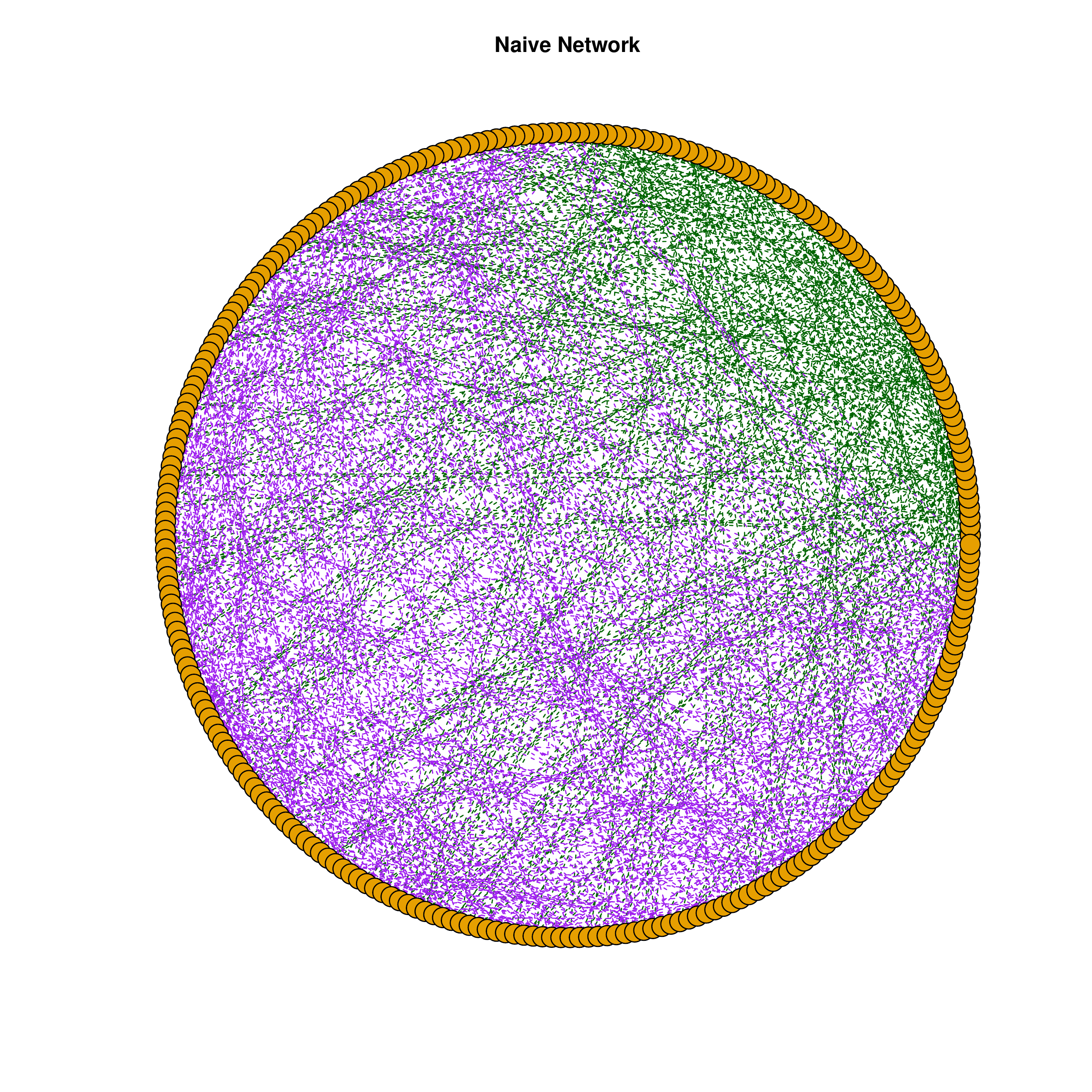} & 
\includegraphics[page = 2 , scale = .4 , trim = {3cm 4cm 4cm 4cm}]{Graph.pdf}
\end{tabular}
\caption{The conditional pair-wise relationships for each of the 273 genes remaining after filtering from the Wilms tumor study.  Each edge represents a conditional pair-wise dependency between two nodes. The left shows the naive analysis, not correcting for measurement error, and the right shows the corrected analysis, correcting for measurement error.  Green edges signify edges found on both graphs, and purple signifies analysis specific edges.}
\label{figure:NetworkGraph}
\end{figure}

\section{Conclusion}

We proposed a correction methodology for Gaussian graphical models when contaminated with additive measurement error.  The core solution to the problem involves using the imputation-regularization algorithm to generate the true values of underlying process with a consistent estimate of the precision matrix.  This provides a consistent, positive-definite estimate of the true precision matrix, which, as simulations illustrate, remove many false positive pair-wise relationships.  Additionally, we show marked improvements in the AUC of the threshold matrix, indicating better separation of the underlying relationships.  From a practitioner point of view, this allows for more reliable downstream analysis and further investigations to be undergone.  

To our knowledge, the novel imputation-regularization algorithm has yet to be used for problems pertaining to contaminated data.  This provides an avenue of future research for more a practical issue in high-dimensional problems, measurement error, which is starting to gain attention.  Moreover, many practical issues still remain in the Gaussian graphical model context, such as the tuning of hyperparameters and the interpretation of the output from the Gibbs sampler-like IRO algorithm. Another potiental avenue of research to pursue is when the amount of measurement error is unknown and not assumed independent.  In this case, sparsity would need to be imposed on $\bm\Omega_u$ in conjunction with $\Omega_x$, posing a challenging, but useful, computational procedure.

\section{Appendix}

\subsection{Proofs}

The proof for Theorem \ref{thm:iro_consistent} in Section \ref{section:ImputationRegularizationME} is established here.  Work done for the IRO algorithm laid the foundation for certain conditions to be met to establish consistency, see the appendix of \cite{liang2018imputation}. We follow closely with their development, and prove the necessary conditions to establish consistency in our context of contaminated GGMs. These conditions include two main parts: (1) the consistency of the regularization step, specifically the BAGUS procedure in our context, and (2) some technical conditions regarding the log-likelihood $\pi(\bmX,\bmW)$. To that end, Assumptions 1 and 2 below ensures the consistency of the BAGUS procedure, while Assumption 3 ensures the metric entropy of the log likelihood not to grow too fast. Discussion of Assumptions 1 and 2 can be found in \cite{gan2018bayesian}, while Assumption 3 has been commonly used in the literature of high-dimensional statistics, see the Remark 1 in the appendix of \cite{liang2018imputation}.

\begin{assumption}
$\lambda_{\text{max}}(\bm\Omega_x) \leq 1 / k_1 \leq \infty$, where $\lambda_{\text{max}}(\bm\Omega_x)$ is the largest eigenvalue of $\bm\Omega_x$ and $k_1$ is a constant such that $k_1 > 0$.
\end{assumption}

For the following assumption we need to define the following values.  Let the column sparsity for $\bm\Omega_x$ be denoted $b = \max_{i = 1,\ldots,d}\sum_{j = 1}^{d} \bm{1}(\omega_{ij} \neq 0)$. For a $m \times q$ matrix $A$ let $\vert \vert \vert A \vert \vert \vert_{\infty} = \max_{1 \leq j \leq q} \sum_{i = 1}^d \vert a_{ij} \vert$ be the maximum absolute row sum. Define $M_{\Sigma_x} = \vert \vert \vert \bm\Sigma \vert \vert \vert_{\infty}$ and $M_{\Gamma} = \vert \vert \vert \bm\Gamma_{s , s}^{-1} \vert \vert \vert_{\infty}$ where $\bm\Gamma = \bm\Sigma_x \otimes \bm\Sigma_x$ and $\bm\Gamma_{s , s}$ denotes the subset of $\bm\Gamma$ by indices $s = \{(i,j) : \bm\Omega_x \neq 0\}$. Let $a_1 > 0$ and $a_2 > 0$ be any predefined constant value.  Also, let $a_3$ and $k_2$ be defined such that $\frac{\log(d)}{n} < a_3 < \frac{1}{4}$ and $\mathbb{E}(e^{tx^{(j)^2}}) \leq k_2$ for all $\vert t \vert \leq a_3$ and $j = 1 , \ldots , d$.  We define $a_4 = a_1(2  + a_2 + a_3^{-1}k_2^2)$, $a_5 = (a_4 + 2M_{\Sigma_x}^2(a_1 + a_4) M_{\Gamma} + 6(a_1 + a_4) b M_{\Gamma}^2 M_{\Sigma}^3 / M$. Finally, define constants $\epsilon_0 > 0$ and $\epsilon_1 > 0$, where $\epsilon_1$ is small.  

\begin{assumption}
For the previously defined constants, the following three statements hold:
\begin{enumerate}
\item The hyperparameters $v_0, v_1, \eta,$ and $\tau$ satisfy
\[\begin{split}
\text{(a) } \quad & \frac{1}{n v_1}  = a_1 \sqrt{\frac{\log(d)}{n}} (1 - \epsilon_0), \\
\text{(b) } \quad & \frac{1}{n v_0} > a_5 \sqrt{\frac{\log(d)}{n}}, \\
\text{(c) } \quad & \frac{v_1^2 (1 - \eta)}{v_0^2 \eta} \leq d^{\epsilon_1}, \\
\text{(d) } \quad & \tau \leq a_1 \frac{n}{2}  \sqrt{\frac{\log(d)}{n}}.
\end{split}\]

\item For the bound $\vert \vert \bm\Omega_x \vert \vert < B$, we have that $B$ satisfies 
\[
\frac{1}{k_1} + 2b(a_1 + a_4)M_{\Gamma}\sqrt{\frac{\log(d)}{n}} < B < \sqrt{2nv_0}.
\]

\item For $M = \max\{2 b (a_1 + a_4) M_{\Gamma} 
\max\{3 M_{\Sigma} , 3 M_{\Gamma} M_{\Sigma}^3 , \frac{2}{k_1^2}  \} , 
\frac{2a_1 \epsilon_0}{k_1^2} \}$, we have $\sqrt{n} \geq M \sqrt{\log(p)}$.
\end{enumerate}

\label{ass:bagus_hyper}
\end{assumption}

\begin{assumption}
The parameter space of $\bm\Omega_x$, or an $L_1$-ball containing the space of $\bm\Omega_x$, grows at a rate of $\mathcal{O}(n^{\alpha})$ for some $0 \leq \alpha \leq \frac{1}{2}$.
\label{ass:metric_entropy}
\end{assumption}

Under these assumptions, we show that the developed procedure to correct for measurement errors satisfy the general conditions for the consistency of the IRO estimate. We state each condition and prove it to hold with our procedure.


\begin{condition}
$\log\pi(\bmX , \bmW \vert \bm\Omega_x)$ is a continuous function of $\bm\Omega_x$ for each $\bmx , \bmw \in \mathbb{R}^d$ and a measurable function of $(\bmX , \bmW)$ for each $\bm\Omega_x$.
\end{condition}
\begin{proof}
We have the expansion
\[
\log\pi(\bmX , \bmW \vert \bm\Omega_x) = \log\pi(\bmX \vert \bm\Omega_x) + \log\pi(\bmW \vert \bmX , \bm\Omega_u).
\]
Hence, the log posterior is continuous for symmetric positive-definite $\bm\Omega_x$ since $\bmx \sim N(\bm0_d , \bm\Omega_x^{-1})$.  The log posterior is also measurable for $(\bmX , \bmW)$ due to properties of the Gaussian distribution.
\end{proof}

\begin{condition}
Three conditions for the Glivenko-Cantelli theorem to hold.
\end{condition}
\begin{enumerate}
\item There exists a function $m_n(\bmX , \bmW)$ such that $\sup_{\bm\Omega_x , \bmX} \vert \log\pi(\bmX , \bmW \vert \bm\Omega_x) \vert \leq m_n(\bmX , \bmW)$.
\item There exists $m_n^*(\bmW)$, such that: 
\begin{enumerate}
\item $0 \leq \int m_n(\bmX , \bmW) \pi(\bmX \vert \bmW , \bm\Omega_x^{(t)}) d\bmX \leq m_n^*(\bmW)$ for all $\bm\Omega_x^{(t)}$,

\item $\mathbb{E}[m_n^*(\bmW)] < \infty$ ,
\item $\sup_{n \in \mathbb{Z}^{+}} \mathbb{E}[m_n^*(\bmW) \mathbb{I}(m_n^*(\bmW) \geq \xi)] \rightarrow 0$ as $\xi \rightarrow \infty$.
\end{enumerate}
Also, as $\xi \rightarrow \infty$, 
\[
\sup_{n \geq 1} \sup_{\bmX , \bm\Omega_x } \vert \int m_n(\bmX , \bmW) \mathbb{I}(m_n(\bmX , \bmW) > \xi) \pi(\bmX \vert \bmW , \bm\Omega_x) \vert \rightarrow 0.
\]

\item  Define $\mathcal{F}_n = \{ \int \log\pi(\bmX , \bmW \vert \bm\Omega_x) \pi(\bmX \vert \bmW , \bm\Omega_x^{(t)}) d\bmX \}$ and $\mathcal{G}_{n,M} = \{q \bm{1}\{m_n^*(\bmW) \leq M\} \vert q \in \mathcal{F}_n\}$. Suppose that, for every $\epsilon$ and $M > 0$, the metric entropy $\log(N(\epsilon, \mathcal{G}_{n , M} , L_1(\mathbb{P}_n))) = \mathcal{O}(n),$ where $\mathbb{P}_n$ is the emprical measure of $\bmW$ and $N(\epsilon, \mathcal{G}_{n , M} , L_1(\mathbb{P}_n))$ is the covering number with respect to the $L_1(\mathbb{P})$-norm.
\end{enumerate}

\begin{proof}
We begin with part (1). Note that 
\[ \begin{split}
\log\pi(\bmX , \bmW \vert \bm\Omega_x) 
& = \sum_{i = 1}^{n} \left[\log\pi(\bmw_i \vert \bmx_i , \bm\Omega_u) + \log\pi(\bmx_i \vert \bm\Omega_x) \right] \\
& =  -\frac{1}{2} \sum_{i = 1}^{n} \left[ (\bmw_i - \bmx_i)^{T} \bm\Omega_u (\bmw_i - \bmx_i) 
+ \bmx_i^T \bm\Omega_x \bmx_i\right] + \frac{1}{2}\log \text{det}(\bm\Omega_x) + C,
\end{split} \]
where $C$ contains constants not related to $(\bmX , \bmW , \bm\Omega_x)$. Hence,
\[ \begin{split}
\vert \log\pi(\bmX , \bmW \vert \bm\Omega_x) \vert
& \leq \frac{1}{2} \sum_{i = 1}^{n} \left[ (\bmw_i - \bmx_i)^{T} \bm\Omega_u (\bmw_i - \bmx_i) 
+  K_1\bmx_i^T\bmx_i\right] + K_2 \\
& = \sum_{i =1}^{n} m(\bmx_i , \bmw_i) = m(\bmX , \bmW),
\end{split} \]
where $K_1$ and $K_2$ are constants depending on upper bound $B$.

To prove part (2) note 
\[ \begin{split}
\tilde{m}(\bmW , \bm\Omega_x^{(t)})
& = \int m(\bmX , \bmW) \pi(\bmX \vert \bmW , \bm\Omega_x^{(t)}) d\bmX \\
& = \int \sum_{i = 1}^{n} m(\bmx_i , \bmw_i) \left[\prod_{j = 1}^{n} \pi(\bmx_i \vert \bmw_i , \bm\Omega_x^{(t)}) \right] d\bmx_1 , \ldots , d\bmx_n \\
& = \sum_{i = 1}^{n} \int m(\bmx_i , \bmw_i) \pi(\bmx_i \vert \bmw_i , \bm\Omega_x^{(t)}) d\bmx_i,
\end{split}\]
where the last equality follows from conditional independence of each $\bmx_i$.  Let $\bm\Lambda^{(t)} = (\bm\Omega_x^{(t)} + \bm\Omega_u)^{-1}$, and notice this the sum of expectations of $m(\bmx_i , \bmw_i)$ with respect to Gaussian random variables following $N(\bm\Lambda^{(t)} \bm\Omega_u \bmw_i , \bm\Lambda^{(t)})$ for each $i = 1 , \ldots , n$. Now,
\[
\mathbb{E}_{\bmx_i \vert \bmw_i , \bm\Omega_x^{(t)}} \left[m(\bmx_i , \bmw_i) \right]
= \frac{1}{2}\bmw_i^{T} \bm\Omega_u \bmw_i + \frac{1}{2}\text{tr}((\bm\Omega_u + K_1 \bm{I}_d)\bm\Lambda^{(t)}) - \underbrace{\bmw_i^{T}\bm\Omega_u\bm\Lambda^{(t)}\bm\Omega_u\bmw_i}_{\geq 0}, 
\]
which, since $\vert \vert \bm\Lambda^{(t)} \vert \vert \leq K_3$, implies
\[
\tilde{m}(\bmW , \bm\Omega_x^{(t)} 
\leq \frac{1}{2} \sum_{i = 1}^{n} \bmw_i^{T} \bm\Omega_u \bmw_i + K_3
= m^*(\bmW).
\]
Marginally $\bmw_i \sim N(\bm0_d , \bm\Sigma_x , \bm\Sigma_u)$, and hence $m^*(\bmW)$ is the sum of scaled chi-square distributions.  Conditions (b) and (c) easily follow from the properties of the chi-square distribution.

To prove part (3), we make use of Remark 1 found in the Appendix of \cite{liang2018imputation}.  Since all elements in $\cup_{n \geq 1} \mathcal{F}_n$ are uniformly Lipschitz, see \cite{honorio2012lipschitz}, the metric entropy can be measured on the basis of the parameter space of $\bm\Omega_x$.  The functions in $\mathcal{G}_{n,M}$ are bounded and the parameter space can be contained by the $L_1$ ball due to the continuity of $\log\pi(\bmX , \bmW \vert \bm\Omega_x)$.  By Assumption \ref{ass:metric_entropy}, then $\log(N(\epsilon, \mathcal{G}_{n , M} , L_1(\mathbb{P}_n))) = \mathcal{O}(n^{2\alpha}\log(d)).$
\end{proof}

\begin{condition}
Define $Z_{t,i} = \log \pi(\bmx_i , \bmw_i \vert \bm\Omega_x) - \int \log \pi(\bmx_i , \bmw_i \vert \bm\Omega_x) \pi(\bmX \vert \bmw_i , \bm\Omega_x^{(t)})$.  $Z_{t,i}$ are subexponential random variables.
\end{condition}

\begin{proof}
First, we note that
\[ \begin{split}
\log\pi(\bmx_i , \bmw_i \vert \bm\Omega_x) & = -\frac{1}{2}(\bmw_i - \bmx_i)^{T} \bm\Omega_u (\bmw_i - \bmx_i) - \frac{1}{2}\bmx_i \bm\Omega_x \bmx_i \\
& = -\frac{1}{2} \bmx_i^{T} (\bm\Omega_x + \bm\Omega_u) \bmx_i + \bmx_i^{T} \bm\Omega_u \bmw_i + C_1,
\end{split} \]
where $C_1$ is a constant free of $\bmX$.  Also note $\log\pi(\bmw_i , \bmX \vert \bm\Omega_x) = \log\pi(\bmw_i , \bmx_i \vert \bm\Omega_x) + \log\pi(\bmX_{-i} \vert \bm\Omega_x)$.  The integral can then be shown to be
\[ \begin{split}
& \int \log \pi(\bmx_i , \bmw_i \vert \bm\Omega_x) \pi(\bmX \vert \bmw_i , \bm\Omega_x^{(t)}) \\
= & \int \big[ \log\pi(\bmw_i , \bmx_i \vert \bm\Omega_x) + \log\pi(\bmX_{-i} \vert \bm\Omega_x) \big]
\pi(\bmx_i \vert \bmw_i , \bm\Omega_x^{(t)}) \pi(\bmX_{-i} \vert \bm\Omega_x^{(t)}) d\bmx_i d\bmX_{-i} \\
= & \underbrace{\int\log\pi(\bmw_i , \bmx_i \vert \bm\Omega_x) \pi(\bmx_i \vert \bmw_i , \bm\Omega_x^{(t)}) d\bmx_i}_{=A} \underbrace{\int  \pi(\bmX_{-i} \vert \bm\Omega_x^{(t)}) d\bmX_{-i}}_{= 1} \\
& \qquad + \underbrace{\int \log\pi(\bmX_{-i} \vert \bm\Omega_x \pi(\bmX_{-i} \vert \bm\Omega_x^{(t)}) d\bmX_{-i}}_{= C_2} \underbrace{\int \pi(\bmx_i \vert \bmw_i , \bm\Omega_x^{(t)}) d\bmx_i}_{= 1}.
\end{split} \]
The value of $A$ is the expectation of $\log\pi(\bmw_i , \bmx_i \vert \bm\Omega_x)$ with respect to the full conditional of $X$ at iteration $t$, $\bmx_i \vert \bmw_i , \bm\Omega_x^{(t)} \sim N_d(\bm\Lambda^{-1 , (t)} \bm{\Omega_u} \bmw_i , \bm\Lambda^{-1 , (t)})$ where $\bm\Lambda^{(t)} = (\bm\Omega_x^{(t)} + \bm{\Omega_u})$.  This expectation is composed of two parts,
\[
\mathbb{E}_{\bmx_i \vert \bmw_i , \bm\Omega_x}(\bmx_i (\bm\Omega_x + \bm\Omega_u) \bmx_i) = 
\text{tr}((\bm\Omega_x + \bm\Omega_u)\Lambda^{-1 , (t)} ) + \bmw_i \bm\Omega_u \bm\Lambda^{(t)} (\bm\Omega_x + \bm\Omega_u) \bm\Lambda^{(t)} \bm\Omega_u \bmw_i
\]
and
\[
\mathbb{E}_{\bmx_i \vert \bmw_i , \bm\Omega_x}(\bmx_i^{T} \bm\Omega_u \bmw_i) = 
\bmw_i^{T} \bm\Omega_u \bm\Lambda^{(t)} \bm\Omega_u \bmw_i.
\]
Hence, $Z_{t,i}$ is 
\[
 -\frac{1}{2} \bmx_i^T (\bm\Omega_x + \bm\Omega_u) \bmx_i + \bmx_i^{T} \bm\Omega_u \bmw_i 
- \frac{1}{2}\bmw_i^T \bm\Omega_u \bm\Lambda^{(t)}(\bm\Omega_x + \bm\Omega_u) \bm\Lambda^{(t)} \bm\Omega_u \bmw_i +\bmw_i^{T} \bm\Omega_u \bm\Lambda^{(t)} \bm\Omega_u \bmw_i + C,
\]
where $C = C_1 + C_2$ is free of $\bmx_i$ and $\bmw_i$, which is the sum of scaled chi-squared distributions and thus is subexponential.
\end{proof}

\begin{condition}
For $t = 1 , \ldots , T$, $Q(\bm\Omega_x \vert \bm\Omega_x^{(t)})$ has a unique maximum at $\tilde{\bm\Omega}_x^{(t)}$; for any $\epsilon > 0$, $\sup_{\bm\Omega_x \backslash B_t(\epsilon)} Q(\bm\Omega_x \vert \bm\Omega_x^{(t)})$ exists, where $B_t(\epsilon) = \{\bm\Omega_x : \vert \bm\Omega_x - \tilde{\bm\Omega}_x^{(t)} \vert < \epsilon\}$.
\end{condition}
\begin{proof}
As noted in \cite{liang2018imputation}, this is satisfied if $\bm\Omega_x$ is restricted to a compact set.  So, since BAGUS is strictly convex when restricted by the condition that $\vert \vert \bm\Omega_x \vert \vert \leq B$, then the condition is satisfied.
\end{proof}

\begin{condition}
The penalty function is non-negative, ensures the existence of $\bm\Omega_x^{(t + 1)}$ for $t = 2 , \ldots , T$, and converges to 0 uniformly as $n \rightarrow \infty$.
\end{condition}
\begin{proof}
BAGUS is a non-negative penalty that exists for any $\bmX$, and, due to the adaptive nature of the penalty, converges to 0 as $n \rightarrow \infty$.  To see the penalty converges to 0, note Assumption \ref{ass:bagus_hyper}.1a implies 
\[
v_1 = \frac{1}{a_1 (1 - \epsilon_0) \sqrt{n\log(d)}} \rightarrow 0
\]
as $n \rightarrow \infty$, which, with a similar argument for $v_0$, results in the penalty being 0 as $n \rightarrow \infty$.
\end{proof}

\subsection{Computing BAGUS with the EM-Algorithm}

Here we review the optimization of the uncontaminated objective distribution.  The direct optimization of $L^{UC}$ in \eqref{eq:trueLogPosterior} is not easy due to the sum inside the logarithm.  \cite{gan2018bayesian} use the EM-algorithm to get around this issue by introducing the latent factors $r_{ij}$ from section \ref{section:BagusVariableSelection}.  This allows an E-step similar to the spike-and-slab Lasso and an M-step similar to the Graphical Lasso.  In this section, if not specified, $\bm\Sigma$ and $\bm\Omega$ refer to $\bmx$'s covariance and precision matrix, respectively. 

The optimization seeks to find the MAP of the posterior proportional to
\[
\vert \bm\Omega_x \vert^{\frac{1}{2}}\exp\left\{-\frac{1}{2}\bmX^T \bm\Omega_x \bmX \right\} 
\prod_{i < j} \pi(\omega_{ij} \vert r_{ij}) \pi(r_{ij} \vert \eta)
\prod_{i} \pi(\omega_{ii} \vert \tau) 
I(\bm\Omega_x \succ 0) I(\vert \vert \bm\Omega_x \vert \vert \leq B),
\]
where the latent indicator $r_{ij}$, as defined in Section \ref{section:BagusVariableSelection}, is incorporated into the off-diagonal elements in the prior specification. The E-step takes the conditional expectation of $r_{ij}$ in the proportional posterior.  Each $r_{ij}$ is conditionally Bernoulli with probability 
\[
p_{ij} = \frac{v_1}{v_0}\frac{1 - \eta}{\eta} \exp\left\{ \vert \omega_{ij}^{(t)} \vert \left(\frac{1}{v_1} - \frac{1}{v_0} \right) \right\},
\]
allowing for easy calculation of the conditional expectation.  Then, the desired $Q$ function to maximize in the M-step is given by 
\[
Q(\bm\Omega_x \vert \bm\Omega_x^{(t)}) = \mathbb{E}_{\bm{R} \vert \bm\Omega_x^{(t)}}\log\pi(\bm\Omega_x , \bmX \vert \bmW , \bm\Sigma_u),
\]
where the expectation is taken element wise for $\bm{R}$ by assumed independence of inclusion.

The M-step optimizes each column of $Q$ seperately with coordinate descent.  The last column's update is now explained, with the other columns following in the same pattern.  Partition the covariance matrix as 
\[
\bm\Sigma_x = \begin{bmatrix}
\bm\Sigma_{11} & \bm\sigma_{12} \\
\bm\sigma_{12}^T & \bm\sigma_{22}
\end{bmatrix},
\]
and let similar partitions be available for $\bm\Omega_x, \bm{P}, \bm{R}, \text{and } \bm{S}$.  Also note that
\[
\begin{bmatrix}
\bm\Sigma_{11} & \bm\sigma_{12} \\
\cdot & \bm\sigma_{22}
\end{bmatrix} = 
\begin{bmatrix}
\bm\Omega_x^{-1} + c^{-1}\bm\Omega_{11}^{-1} \bm\omega_{12} \bm\omega_{12}^{T} \bm\Omega_{11}^{-1} &
-c^{-1} \bm\Omega_{11}^{-1} \bm\omega_{12} \\
\cdot & c^{-1}
\end{bmatrix},
\]
where $c = \omega_{22} - \bm\omega_{12}^{T} \bm\Omega_{11}^{-1} \bm\omega_{12}.$  The update for the last column of $\bm\Sigma_x$ is the solution from setting subgradient of $Q$ with respect to $\left[\bm\sigma_{12} \ \bm\sigma_{22}\right]^{T}$ to 0.  The update for $\sigma_{22}$ can is easily attained from the setting the subgradient of $\omega_{22}$ to 0, 
\[
\omega_{22} = \frac{1}{\sigma_{22}} + \bm\omega_{12}^{T} \bm\Omega_{11}^{-1} \bm\omega_{12}.
\]
We note that each column update requires the matrix $\bm\Omega_{11}^{-1}$.  This can be computed as $\bm\Sigma_{11} - \bm\sigma_{12} \bm\sigma_{12}^{T} / \sigma_{22}$. 

\bibliography{references}
\bibliographystyle{plain}

\end{document}